\documentclass[aip,jmp]{revtex4-1}

\usepackage{latexsym}
\usepackage[T1]{fontenc}
\usepackage[utf8x]{inputenc}
\usepackage{times}
\usepackage{amsmath}
\usepackage{amssymb}
\usepackage{amscd}
\usepackage{amsfonts}
\usepackage{amstext}
\usepackage{amsthm}
\usepackage{graphicx}
\usepackage[all]{xy}
\usepackage{float}
\usepackage{multirow}
\usepackage{color}
\newcommand{\hq }{{/\kern -.185em/}}

\theoremstyle{plain}

\newtheorem{theorem} {Theorem}
\newtheorem{lemma} [theorem]{Lemma}
\newtheorem{proposition}[theorem]{Proposition}

\theoremstyle{definition}

\newcommand{\ket}[1]{\mbox{$| #1 \rangle$}}
\newcommand{\bk}[2]{\ensuremath{\langle #1 | #2 \rangle}}

\begin{document}

\title{How many invariant polynomials are needed to decide local unitary equivalence of qubit states? }

\author{Tomasz Maci\c{a}\.{z}ek}
\affiliation{Center for Theoretical Physics, Polish Academy of
Sciences, Al. Lotnik\'ow 32/46, 02-668 Warszawa, Poland}
\affiliation{Faculty of Physics, University of Warsaw, ul. Ho\.{z}a 69, 00-681 Warszawa, Poland}
\author{Micha{\l} Oszmaniec}
\affiliation{Center for Theoretical Physics, Polish Academy of
Sciences, Al. Lotnik\'ow 32/46, 02-668 Warszawa, Poland}
\author{Adam Sawicki}
\affiliation{Center for Theoretical Physics, Polish Academy of
Sciences, Al. Lotnik\'ow 32/46, 02-668 Warszawa, Poland}
\affiliation{School of Mathematics, University of Bristol, University Walk, Bristol BS8 1TW, UK}

\begin{abstract}
Given $L$-qubit states with the fixed spectra of reduced one-qubit density matrices, we find a formula for the minimal number
of invariant polynomials needed for solving local unitary (LU) equivalence problem, that is, problem of deciding if two states can be connected by local unitary operations. Interestingly, this number is not the same for every collection of the spectra. Some spectra require less polynomials to solve LU equivalence problem than others. The result is obtained using geometric methods, i.e. by calculating the dimensions of {\it reduced spaces}, stemming from the symplectic reduction procedure.
\end{abstract}

\maketitle

\section{Introduction }
\label{sec:intro}

We consider a quantum system consisting of $L$ isolated qubits with the Hilbert space $\mathcal{H}=\left(\mathbb{C}^{2}\right)^{\otimes L}$.
We assume that states are normalized to one and neglect a global phase. In this way the space of pure states is isomorphic to
the complex projective space $\mathbb{P}(\mathcal{H})$. In the following, by $[\phi]\in\mathbb{P}(\mathcal{H})$ we will
denote a state corresponding to a vector $\phi\in\mathcal{H}.$ Each qubit is located in a different laboratory and the available
operations are restricted to the local unitaries described by the group $K=SU(2)^{\times L}$. Two states are called locally unitary equivalent
(LU equivalent) if and only if they can be connected by the action of $K$, that is, belong to the same $K$-orbit. The problem of local unitary equivalence of states can be in principle solved by finding the set of $K$-invariant polynomials, i.e. polynomials that are constant on $K$-orbits (see \cite{Kraus1,Kraus2} for another approach). When the number of qubits is large this, however, becomes hard as the number of polynomials grows exponentially with the
number of constituents of the system. The problem of LU equivalence for bipartite and three-qubit pure
states was recently studied form the symplecto-geometric perspective \cite{SK11,SWK13} (see also \cite{ZM12}). In particular, the connection with the symplectic reduction
was established. The current paper can be seen as a generalization of these ideas to an arbitrary
number of qubits.

Among the $K$-invariant polynomials there are $L$ polynomials $\left\{ \mathrm{tr}\left(\rho_{l}^{2}([\phi])\right)\right\} _{l=1}^{L}$,
where $\rho_{l}([\phi])$ are the reduced one-qubit density matrices. Consequently, for two LU equivalent states $[\phi_{1,2}]\in\mathbb{P}(\mathcal{H})$
the spectra of the corresponding reduced one-qubit density matrices are the same. If we denote by $\Psi$ the map which assigns
to $[\phi]\in\mathbb{P}(\mathcal{H})$ the shifted spectra of its reduced one-qubit density matrices, i.e. $\Psi([\phi])=\{\mathrm{diag}(-\lambda_1,\lambda_1),\ldots,\mathrm{diag}(-\lambda_L,\lambda_L)\}$, where $\lambda_i=\frac{1}{2}-p_i$ and $\{p_i,1-p_i\}$ is the increasingly ordered spectrum of $\rho_i([\phi])$, then the states satisfying the
above necessary condition form a fiber of $\Psi$. Fibers of $\Psi$ are connected collections of $K$-orbits \cite{Kirwan-thesis}. Moreover, the image, $\Psi(\mathbb{P}(\mathcal{H}))$, is a convex polytope \cite{K84}. The polynomials $\left\{ \mathrm{tr}\left(\rho_{l}^{2}([\phi])\right)\right\} _{l=1}^{L}$ restricted to a fiber of $\Psi$ are constant functions. Therefore, typically, for two states $[\phi_{1}]$ and $[\phi_{2}]$ with $\alpha:=\Psi([\phi_{1}])=\Psi([\phi_{2}])$, where $\alpha$ denote the collection of spectra of one-qubit density matrices, some additional $K$-invariant polynomials are needed to decide the LU equivalence. The number of these polynomials is given by the dimension of the {\it reduced space} $M_{\alpha}:=\Psi^{-1}\left(\alpha\right)/K$ (see \cite{SK11,SWK13}). Interestingly, $\mathrm{dim}M_{\alpha}$
may not be the same for every $\alpha\in\Psi(\mathbb{P}(\mathcal{H}))$, that is, some collections of spectra of reduced
one-qubit density matrices require less additional polynomials to solve LU equivalence problem than others. In particular,
if the fiber $\Psi^{-1}\left(\alpha\right)$ contains exactly one $K$-orbit, i.e. $\mathrm{dim}M_{\alpha}=0$, no additional
information is needed and any two states $[\phi_{1,2}]\in\Psi^{-1}\left(\alpha\right)$ are LU equivalent.

In this paper we find the formula for the dimension of the reduced space $M_{\alpha}=\Psi^{-1}(\alpha)/K$, for any $\alpha\in\Psi(\mathbb{P}(\mathcal{H}))$ and for an arbitrary number $L$ of qubits. Our result is obtained in two steps. First, we consider the points $\alpha_{gen}\in\Psi(\mathbb{P}(\mathcal{H}))$
which belong to the interior of the polytope $\Psi(\mathbb{P}(\mathcal{H}))$. In this case the map $\Psi$ is {\it regular}
and the calculation is rather straightforward. The dimension of $M_{\alpha_{gen}}$ does not depend on $\alpha_{gen}$. Moreover, $\mathrm{dim}M_{\alpha_{gen}}+L$
is equal to the cardinality of the spanning set of $K$-invariant polynomials. For points $\alpha_{b}\in\Psi(\mathbb{P}(\mathcal{H}))$
which belong to the boundary of $\Psi(\mathbb{P}(\mathcal{H}))$ the problem requires more advanced methods and turns out
to be more interesting. In particular, for a large part of the boundary of $\Psi(\mathbb{P}(\mathcal{H}))$ we have $\mathrm{dim}M_{\alpha_{b}}=0$. We also observe
that for $\alpha_{b}\in\Psi(\mathbb{P}(\mathcal{H}))$ corresponding to the $\left\{ \rho_{l}([\phi])\right\} _{l=1}^{L}$ such that $k$ matrices are
maximally mixed $\mathrm{dim}M_{\alpha_{b}}=\mathrm{dim}M_{\alpha_{gen}}-2k$.

\section{How many invariant polynomials are needed to decide LU equivalence of 4-qubit states? }
\label{sec:four-qubits}
In this section we briefly discuss the considered problem and present the main results of the paper on the 4 qubits example.

Recently, the problem of finding $\mathrm{dim}M_{\alpha}$ was considered for three qubits \cite{SWK13}. In particular it was shown
that for points in the interior of the polytope $\Psi(\mathbb{P}(\mathcal{H}))$, $\mathrm{dim}M_{\alpha}=2$, whereas for
points on the boundary $\mathrm{dim}M_{\alpha}=0$. The uniform behaviour of $\mathrm{dim}M_{\alpha}$ on the boundary of
$\Psi(\mathbb{P}(\mathcal{H}))$ in case of three qubits is, as already indicated in ref.\cite{SWK13}, a low dimensional phenomenon.
As we explain in section \ref{sec:The-reduced-spaces}, for an arbitrary number of $L$ qubits the boundary consists of three parts characterized by a different behaviour of $\mathrm{dim}M_{\alpha}$. The first part is the polytope $\Psi(\mathbb{P}(\tilde{\mathcal{H}}))$ that corresponds to a system with one qubit less, that is, $\tilde{\mathcal{H}}=\left(\mathbb{C}^2\right)^{\otimes (L-1)}$. The second corresponds to changing one of the non-trivial inequalities (\ref{nier1}) into an equality. The third represents situations when $k$ one-qubit density matrices are maximally mixed. The clear distinction between these parts of the boundary can be seen already in the four qubits case.

The four-qubit polytope $\Psi(\mathbb{P}(\mathcal{H}))$ is a $4$-dimensional convex polytope spanned by 12 vertices (see
appendix for the proof and the list of vertcies). The dimension of the reduced space in the interior of $\Psi(\mathbb{P}(\mathcal{H}))$ is $\textrm{dim}M_{\alpha_{gen}}=14$ (see formula (\ref{eq:fiber-interior-1})). In figure \ref{fig:Three-kinds-of} the above mentioned three different parts of the boundary are shown. In particular, in figure \ref{fig:Three-kinds-of}(a) we see that for $3$-dimensional face of $\Psi(\mathbb{P}(\mathcal{H}))$ corresponding to three qubits, $\mathrm{dim}M_{\alpha}=2$ in the interior and $\mathrm{dim}M_{\alpha}=0$ on the boundary which agrees with results of ref. \cite{SWK13}. On the other hand, inside the $3$-dimensional face shown in figure \ref{fig:Three-kinds-of}(b) corresponding to one of $\{\rho_{i}\}_{i=1}^{4}$ being
maximally mixed we have $\mathrm{dim}M_{\alpha}=12$. The boundary of this face contains: $2$-dimensional faces corresponding to
two of $\{\rho_{i}\}_{i=1}^{4}$ being maximally mixed - $\mathrm{dim}M_{\alpha}=10$, $1$-dimensional faces - three of
$\{\rho_{i}\}_{i=1}^{4}$ are maximally mixed and $\mathrm{dim}M_{\alpha}=8$, and finally, the vertex denoted by $v_{\mathrm{GHZ}}$
when all one-particle reduced density matrices are maximally mixed - $\mathrm{dim}M_{\alpha}=6$. Therefore, as mentioned in the introduction,
$\mathrm{dim}M_{\alpha_{b}}=\mathrm{dim}M_{\alpha_{gen}}-2k$. Finally, in figure \ref{fig:Three-kinds-of}(c) we see the
$3$-dimensional face of $\Psi(\mathbb{P}(\mathcal{H}))$ with $\mathrm{dim}M_{\alpha}=0$.

\begin{figure}[H]
\centering
\includegraphics[width=14cm]{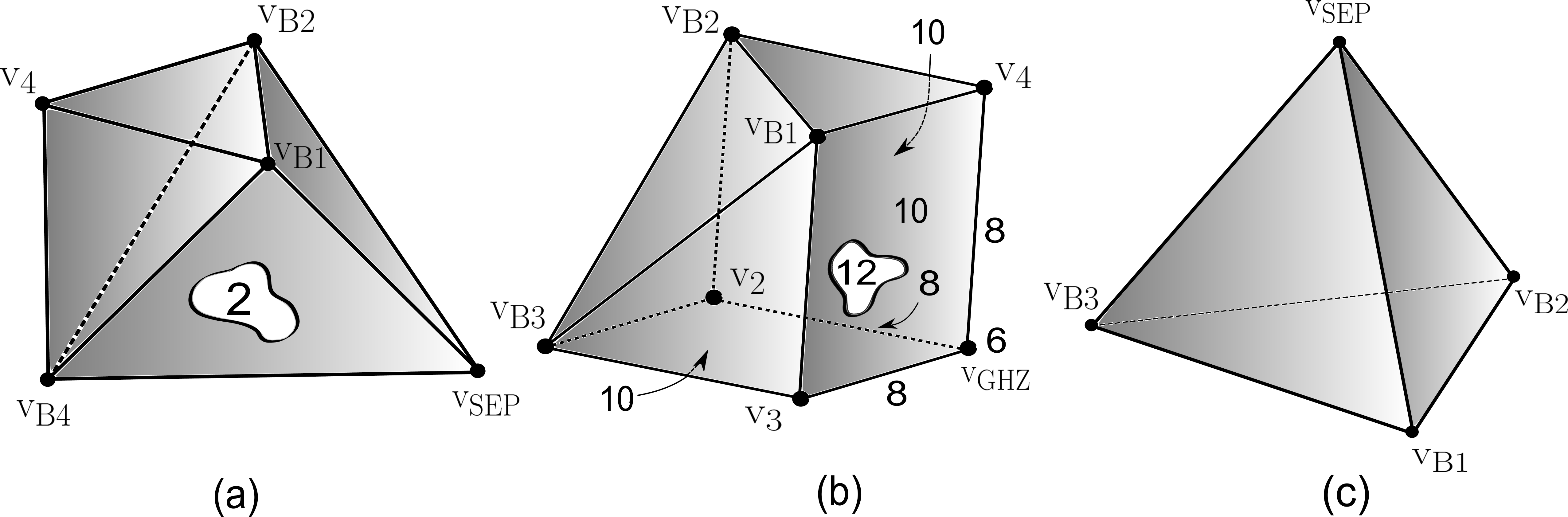}
\caption{\label{fig:Three-kinds-of}Three parts of the boundary of $\Psi(\mathbb{P}(\mathcal{H}))$. The numbers denote $\mathrm{dim}M_{\alpha}$.
If the number is missing, $\mathrm{dim}M_{\alpha}=0$. The vertices are defined in appendix.}
\end{figure}

\noindent In the next sections we show how to calculate $\mathrm{dim}M_{\alpha}$ for any point of $\Psi(\mathbb{P}(\mathcal{H}))$
for an arbitrary number of qubits.

\section{LU equivalence of qubits and the reduced spaces $M_{\alpha}$}
\label{sec:The-reduced-spaces}

 We start with a rigorous statement of the problem and the solution of its easy part. For a detailed description of symplecto-geometric methods in quantum information theory see for example \cite{CDKW12,SHK11,SK11,SWK13,WDGC12}.

Let $\mathcal{H}=\left(\mathbb{C}^{2}\right)^{\otimes L}$ be the $L$-qubit Hilbert space and denote by $\mathbb{P}(\mathcal{H})$
the corresponding complex projective space. It is known that $\mathbb{P}(\mathcal{H})$ is a symplectic manifold with the
Fubiny-Study symplectic form $\omega_{FS}$. The action of $K=SU(2)^{\times L}$ on $\mathbb{P}(\mathcal{H})$ is symplectic,
i.e. it preserves $\omega_{FS}$. Consequently, there is the {\it momentum map} for this action. In the considered setting
this map is given by
\[\mu:\mathbb{P}(\mathcal{H})\rightarrow i\mathfrak{k}\]
\[\mu([\phi])=\left\{ \rho_{1}([\phi])-\frac{1}{2}I,\ldots,\rho_{L}([\phi])-\frac{1}{2}I\right\} , \]
where $\mathfrak{k}$ is the Lie algebra of $K$, $\left\{ \rho_{l}([\phi])\right\} _{l=1}^{L}$ are the reduced one-qubit
density matrices, $I$ is $2\times2$ identity matrix and $i^{2}=-1$. The map $\mu$ is equivariant, i.e. for any $g\in K$
\begin{gather}\label{equvariance}
\mu([g\phi])=\left\{ \rho_{1}([g\phi])-\frac{1}{2}I,\ldots,\rho_{L}([g\phi])-\frac{1}{2}I\right\} = \\ =\left\{ g\left(\rho_{1}([\phi])-\frac{1}{2}I\right)g^{\ast},\ldots,g\left(\rho_{L}([\phi])-\frac{1}{2}I\right)g^{\ast}\right\}
=g\mu([\phi])g^{\ast}=\mathrm{Ad}_{g}\mu([\phi]),
\end{gather}
where $g^\ast$ is the Hermitian conjugate of $g$ and  $\mathrm{Ad}_{g}X:=gXg^{\ast}$, for any $X\in i\mathfrak{k}$. The equivariance of $\mu$ implies that $K$-orbits in $\mathbb{P}(\mathcal{H})$ are mapped onto adjoint orbits in $i\mathfrak{k}$.
Consequently, the necessary condition for two states $[\phi_{1,2}]\in\mathcal{H}$ to be on the same $K$-orbit in $\mathbb{P}(\mathcal{H})$
is $\mu(K.[\phi_{1}])=\mu(K.[\phi_{2}])$. On the other hand, adjoint orbits in $i\mathfrak{k}$ are determined by the spectra
of matrices $\left\{ \rho_{l}([\phi])\right\} _{l=1}^{L}$. The characteristic polynomial $w_{l}(\nu)$
for $\rho_{l}([\phi])$ reads

\begin{gather}
w_{l}(\nu)=\nu^{2}-\mathrm{tr}\left(\rho_{l}([\phi])-\frac{1}{2}I\right)\nu+\left(\left(\mathrm{tr}\left(\rho_{l}([\phi])-\frac{1}{2}I\right)\right)^2-\mathrm{tr}\left(\rho_{l}([\phi])-\frac{1}{2}I\right)^{2}\right).\label{eq:ch-poly}
\end{gather}
Using $\mathrm{tr}\left(\rho_{l}([\phi])\right)=1$, equation (\ref{eq:ch-poly}) reduces to

\begin{gather}
w_{l}(\nu)=\nu^{2}-\mathrm{tr}\left(\rho_{l}^{2}([\phi])\right)-\frac{1}{2}.\label{eq:ch-poly-1}
\end{gather}
One can thus say that the necessary condition for two states $[\phi_{1,2}]\in\mathbb{P}(\mathcal{H})$ to be on the same
$K$-orbit in $\mathbb{P}(\mathcal{H})$ is that

\begin{gather*}
\forall l\,\,\mathrm{tr}\left(\rho_{l}^{2}([\phi_{1}])\right)=\mathrm{tr}\left(\rho_{l}^{2}([\phi_{2}])\right).
\end{gather*}
Let $i\mathfrak{t}_{+}$ be the positive Weyl chamber in $i\mathfrak{k}$, i.e.

\begin{gather}
i\mathfrak{t}_{+}=\left\{ \left(\begin{array}{cc}
-\lambda_{1} & 0\\
0 & \lambda_{1}
\end{array}\right),\ldots,\left(\begin{array}{cc}
-\lambda_{L} & 0\\
0 & \lambda_{L}
\end{array}\right)\,:\,\lambda_{i}\in\mathbb{R}_{+}\right\} .\label{eq:chamber}
\end{gather}
We define the map $\Psi:\mathbb{P}(\mathcal{H})\rightarrow i\mathfrak{t}_{+}$ to be

\begin{gather*}
\Psi([\phi]):=\mu(K.[\phi])\cap i\mathfrak{t}_{+}=\left\{ \tilde{\rho}_{1}([\phi])-\frac{1}{2}I,\ldots,\tilde{\rho}_{L}([\phi])-\frac{1}{2}I\right\} ,
\end{gather*}
where each $\tilde{\rho}_{l}([\phi])$ is a diagonal $2\times2$ matrix whose diagonal elements are given by the increasingly
ordered spectrum $\sigma\left(\rho_{l}([\phi]\right)=\{p_{l},\,1-p_{l}\}$ of $\rho_{l}([\phi])$, that is, $\frac{1}{2}\geq p_{l}\geq0$. The shifted spectrum, i.e. the spectrum of $\rho_{l}([\phi])-\frac{1}{2}I$ is given by $\{-\lambda_{l},\lambda_{l}\}$,
where
\begin{gather*}
0\leq\lambda_{l}=\frac{1}{2}-p_{l}\leq\frac{1}{2}.
\end{gather*}
Under these assumptions the image $\Psi(\mathbb{P}(\mathcal{H}))$ is known to be a convex polytope, defined by the following
set of inequalities \cite{HSS03}
\begin{gather}
\forall_{l}\ 0\leq\lambda_{l}\leq\frac{1}{2},\,\,\mathrm{and}\,\,\left(\frac{1}{2}-\lambda_{l}\right)\leq\sum_{j\neq l}\left(\frac{1}{2}-\lambda_{j}\right)\,.\label{nier1}
\end{gather}
Note that the polytope $\Psi(\mathbb{P}(\mathcal{H}))$ gives 1-1 parametrization of $\mu(\mathbb{P}(\mathcal{H}))$, that is, of all available adjoint $K$-orbits. Moreover, for each point $\alpha\in\Psi(\mathbb{P}(\mathcal{H}))$ the set $\Psi^{-1}(\alpha)$
is a connected $K$-invariant stratified symplectic space \cite{MW99}. We will next briefly describe the structure of this space \cite{MW99}. Note that isotropy groups of points belonging to a fixed $K$-orbit are conjugated.
First, one can decompose $\mathbb{P}(\mathcal{H})$ according to the conjugacy classes of isotropy groups, i.e. into sets
of points, whose isotropies are conjugated. This decomposition divides $\mathbb{P}(\mathcal{H})$ into sets of orbits characterized by the same isotropy type. These sets are in general non-connected, thus performing the decomposition further into connected components, one obtains the stratification
\[
\mathbb{P}(\mathcal{H})=\cup S_{\nu}.
\]
There exists an unique highest dimensional stratum $S^{max}$, which is open and dense. Moreover, it can be shown that for
each $\alpha\in\Psi(\mathbb{P}(\mathcal{H}))$ the set $\mu^{-1}(\alpha)\cap S_{\nu}$ is smooth and that the decomposition
\begin{gather}
\mu^{-1}(\alpha)=\cup_{\nu}\mu^{-1}(\alpha)\cap S_{\nu},
\label{eq:strat}
\end{gather}
gives the stratification of $\mu^{-1}(\alpha)$. Let  $K_{\alpha}$ be the isotropy subgroup of $\alpha$ with respect to the adjoint action, i.e. $K_{\alpha}=\{g\in K:\,\mathrm{Ad}_{g}\alpha=\alpha\}$.  Because $\mu^{-1}(\alpha)$ is invariant to the action of $K_{\alpha}$, one has that the \textit{reduced space} $M_{\alpha}=\Psi^{-1}(\alpha)/K\cong\mu^{-1}(\alpha)/K_{\alpha}$.
The stratification given by (\ref{eq:strat}) induces the stratification of $M_{\alpha}$ in a natural way. What is more, every $M_{\alpha}$ also has an unique highest dimensional stratum $M_{\alpha}^{max}$. If $\mu^{-1}(\alpha)\cap S^{max}\neq\emptyset$ then
\[
M_{\alpha}^{max}=(\mu^{-1}(\alpha)\cap S^{max})/K.
\]
Therefore it is clear that the dimension of the highest dimensional stratum in $M_{\alpha}$, or more precisely $\mathrm{dim}M_{\alpha}+L$
is the number of $K$-invariant polynomials needed to decide LU equivalence of states satisfying $\Psi([\phi])=\alpha$. The problem of finding $\mathrm{dim}M_{\alpha}$ is essentially different for points $\alpha$ in the interior and on the boundary of $\Psi(\mathbb{P}(\mathcal{H}))$. Using (\ref{nier1}) the boundary of $\Psi(\mathbb{P}(\mathcal{H}))$ can be divided into three parts
\begin{description}
  \item[\textbf{Case 1}] $k$ of $\lambda_{l}$s are equal to $\frac{1}{2}$.
  \item[\textbf{Case 2}] At least one of the inequalities $\left(\frac{1}{2}-\lambda_{l}\right)\leq\sum_{j\neq l}\left(\frac{1}{2}-\lambda_{j}\right)$ is an equality.
  \item[\textbf{Case 3}] $k$ of $\lambda_{l}$s are equal to $0$.
\end{description}
By the permutation symmetry of inequalities (\ref{nier1}), it is enough to consider one example for each of these cases. In the remaining two paragraphs of this section we find $\mathrm{dim}M_\alpha$ for the interior and the boundary points satisfying the condition of the first case.

\subsection{Interior of $\Psi(\mathbb{P}(\mathcal{H}))$ }

For the points $\alpha$ in the interior of the polytope $\Psi(\mathbb{P}(\mathcal{H}))$ the calculation of $\mathrm{dim}M_{\alpha}$
turns out to be rather straightforward.

To begin, let us denote by $\mathbb{P}(\mathcal{H})^{max}$ the union of $K$-orbits of maximal dimension in $\mathbb{P}(\mathcal{H})$.
The principal isotropy theorem \cite{B72} implies that this set is connected, open and dense. In order to calculate the
maximal dimension of $K$-orbits in $\mathbb{P}(\mathcal{H})$ we use the following fact (see ref. \cite{SWK13})
\begin{proposition}
\label{dim-orbit}
\label{dim-orbit}Assume that $\mathrm{dim}\Psi(\mathbb{P}(\mathcal{H}))=\mathrm{dim}\mathfrak{t}_{+}$. Then a generic $K$-orbit
has dimension of the group $K$.
\end{proposition}
\noindent For $L$-qubits, $L\geq3$, one easily checks that the polytope $\Psi(\mathbb{P}(\mathcal{H}))$ given by inequalities
(\ref{nier1}), is $L$-dimensional. On the other hand, for the Weyl chamber defined by (\ref{eq:chamber}), we have $\mathrm{dim}\mathfrak{t}_{+}=L$.
Therefore, by proposition \ref{dim-orbit} the $K$-orbits belonging to $\mathbb{P}(\mathcal{H})^{max}$ have dimension of the group
$K=SU(2)^{\times L}$, i.e $3L$. Moreover, $\Psi(\mathbb{P}(\mathcal{H})^{max})$ contains the interior of the polytope
$\Psi(\mathbb{P}(\mathcal{H}))$ (see ref.\cite{HH96}) and for $\alpha$ in the interior of the polytope $\Psi(\mathbb{P}(\mathcal{H}))$
the set $\Psi^{-1}(\alpha)\cap\mathbb{P}(\mathcal{H})^{max}$ is the highest dimensional stratum of $\Psi^{-1}(\alpha)$
(see ref.\cite{MW99}). Hence
\begin{gather}
\mathrm{dim}M_{\alpha}=\mathrm{dim}(\Psi^{-1}(\alpha)/K)=\left(\mathrm{dim}\mathbb{P}(\mathcal{H})-\mathrm{dim}\Psi(\mathcal{H})\right)-\mathrm{dim}K=\nonumber \\
=\left(\left(2^{L+1}-2\right)-L\right)-3L=2^{L+1}-4L-2.\label{eq:fiber-interior-1}
\end{gather}
Note that already in the case of three qubits we have $\mathrm{dim}M_{\alpha}=2$, that is, one needs $3+2=5$, $K$-invariant polynomials
to decide LU equivalence of states whose spectra of reduced density matrices are the same and belong to the interior of $\Psi(\mathbb{P}(\mathcal{H}))$.
The exponential growth of $\mathrm{dim}M_{\alpha}$, which for large $L$ is of order $2^{L}$, can be seen as a usual statement
that the number of $K$-invariant polynomials needed to distinguish between generic $K$-orbits grows exponentially with
the number of particles.

\subsection{Calculation of $\mathrm{dim}M_{\alpha}$ for case 1}

\noindent The Case 1 is once again straightforward. To see this note that when the first $k$ of $\lambda_{l}$'s are equal to $\frac{1}{2}$
inequalities (\ref{nier1}) reduce to the inequalities for the polytope of $L-k$ qubits. Moreover, a generic state belonging
to $\Psi^{-1}(\alpha)$ is of the form $\phi_{1}\otimes\phi_{2}$ where $\phi_{1}$ is a $k$-qubit separable state and $\phi_{2}$
is an arbitrary state of $L-k$ qubits. Therefore, using (\ref{eq:fiber-interior-1}), we get

\begin{gather*}
\mathrm{dim}M_{\alpha}=\left(\left(2^{L-k+1}-2\right)-(L-k)\right)-3(L-k)=2^{L-k+1}-4(L-k)-2.
\end{gather*}
In the next two sections we find $\mathrm{dim}M_{\alpha}$ for cases 2 and 3.

\section{Calculation of $\mathrm{dim}M_{\alpha}$ for case 2 }
\label{sec:case2}

In this section we show that $\mathrm{dim}M_{\alpha}=0$ for $\alpha$'s satisfying assumptions of case 2.

Assume that one of the inequalities $\left(\frac{1}{2}-\lambda_{l}\right)\leq\sum_{j\neq l}\left(\frac{1}{2}-\lambda_{j}\right)$
is an equality, e.g.
\begin{gather}
-\lambda_{1}+\sum_{i=2}^{L}\lambda_{i}=\frac{1}{2}L-1.\label{eq:condition1-1-1}
\end{gather}
Using the fact that states mapped by $\Psi$ onto $\alpha$ are $K$-orbits through states mapped by $\mu$ onto $\alpha$, i.e. $\Psi^{-1}(\alpha)/K=\mu^{-1}(\alpha)/K_{\alpha}$, one can, if more convenient,  use map $\mu$ instead of $\Psi$ to calculate $\mathrm{dim}M_{\alpha}$. It turns out that this is the case for the considered $\alpha$'s, as we have the following:
\begin{proposition}
\label{prop1}
Let $\phi\in\mathcal{H}$ be such that $\mu([\phi])$ satisfies (\ref{eq:condition1-1-1}). Let $\xi=[\xi_{1},\ldots,\xi_{L}]$ be a vector perpendicular to the plane given by (\ref{eq:condition1-1-1}). Then $\phi$ is an eigenvector of $X=X_{1}\otimes I\otimes\ldots\otimes I+\ldots+I\otimes I\otimes\ldots\otimes X_{L}$,
where $X_{l}=\mathrm{diag}\{\xi_{l},-\xi_{l}\}$.
\end{proposition}
In order to prove this, we will use the fact characterizing the image of the differential of the momentum map \cite{GS84}:
\begin{proposition}
\label{prop2}
The image of $d\mu|_{[\phi]}:T_{[\phi]}M\rightarrow i\mathfrak{k}$ is equal to the annihilator of $\mathfrak{k}_{[\phi]}$,
the Lie algebra of the isotropy subgroup $K_{[\phi]}\subset K$.
\end{proposition}
\begin{proof}
(Proposition \ref{prop1}) The Lie algebra of $K$ is a real vector space equipped with the inner product given by $\bk AB=-\frac{1}{2}\mathrm{tr}(AB)$.
The matrices
\[\mathcal{X}_{k}=i\mathbb{I}\otimes\mathbb{I}\otimes...\otimes\sigma_{x}\otimes\mathbb{I}\otimes...\otimes\mathbb{I},\]
\[\mathcal{Y}_{k}=i\mathbb{I}\otimes\mathbb{I}\otimes...\otimes\sigma_{y}\otimes\mathbb{I}\otimes...\otimes\mathbb{I},\]
\[\mathcal{Z}_{k}=i\mathbb{I}\otimes\mathbb{I}\otimes...\otimes\sigma_{z}\otimes\mathbb{I}\otimes...\otimes\mathbb{I},\]
where $\sigma_{x,y,z}$ are Pauli matrices, form an orthogonal basis of $\mathfrak{k}$. The matrix representing the collection
$\left\{ \rho_{1}([\phi])-\frac{1}{2}I,\ldots,\rho_{L}([\phi])-\frac{1}{2}I\right\} $ is of the form:
\begin{gather*}
(\rho_{1}([\phi])-\frac{1}{2}I)\otimes I\otimes\ldots\otimes I+\ldots+I\otimes\ldots\otimes I\otimes(\rho_{L}([\phi])-\frac{1}{2}I).
\end{gather*}
Note that $d\mu|_{[\phi]}$ is the map that transports vectors tangent to $\mathbb{P}(\mathcal{H})$
at the point $[\phi]$ to the tangent space to $\mathfrak{k}$, $T_{\mu([\phi])}\mathfrak{k}\cong\mathfrak{k}$. Assume that $\alpha:=\mu([\phi])$ belongs to
$\Psi(\mathbb{P}(\mathcal{H}))$ and that spectra of the corresponding matrices $\{\rho_{l}([\phi])\}_{i=1}^{L}$ are nondegenerate.

In this case the Lie algebra $\mathfrak{k}_{\alpha}$ of the isotropy subgroup $K_{\alpha}\subset K$ is given by the diagonal
matrices in $\mathfrak{k}$, the set of which we denote by $\mathfrak{t}$. On the other hand by the equivariance of the momentum map (\ref{equvariance})
we have $\mathfrak{k}_{[\phi]}\subset\mathfrak{k}_{\alpha}=\mathfrak{t}$. In other words, states $[\phi]$ mapped by $\mu$
to $\alpha\in\Psi(\mathbb{P}(\mathcal{H}))$ with non-degenerate spectra of $\{\rho_{l}([\phi])\}_{i=1}^{L}$ can have isotropy
given by at most $\mathfrak{t}$. As the off-diagonal matrices in the image of $d\mu|_{[\phi]}$ are orthogonal to $\mathfrak{t}$,
by proposition \ref{prop2}, in order to find $\mathfrak{k}_{[\phi]}$ we need to find matrices in $\mathfrak{t}$ orthogonal to
(annihilated by) diagonal matrices from $d\mu|_{[\phi]}$. Note that since $\Psi(\mathbb{P}(\mathcal{H}))\subset\mathfrak{t}$
and $\mathrm{dim}\Psi(\mathbb{P}(\mathcal{H}))=\mathrm{dim}\mathfrak{t}$, for $\alpha$ inside $\Psi(\mathbb{P}(\mathcal{H}))$
the diagonal matrices in $T_{\mu([\phi])}\mathfrak{k}$, for $[\phi]\in\mathbb{P}(\mathcal{H})^{max}$, span the space $\mathfrak{t}$
and hence the isotropy $\mathfrak{k}_{[\phi]}=0$. On the other hand, for $\alpha$ satisfying (\ref{eq:condition1-1-1})
the space $T_{\mu([\phi])}\mathfrak{k}\cap\mathfrak{t}\neq\mathfrak{t}$. In order to find matrices in $\mathfrak{t}$ orthogonal
to $T_{\mu([\phi])}\mathfrak{k}\cap\mathfrak{t}$ note that any element of $\mathfrak{t}$ can be written as $\sum_{k=1}^{L}a_{k}\mathcal{Z}_{k}$.
The inner product of two matrices of this type reads

\[
\bk{\sum_{k=1}^{L}a_{k}\mathcal{Z}_{k}}{\sum_{k=l}^{L}b_{l}\mathcal{Z}_{l}}=-\frac{1}{2}\sum_{k=1}^{L}a_{k}\sum_{k=l}^{L}b_{l}\mathrm{tr}\left(\mathcal{Z}_{k}\mathcal{Z}_{l}\right)=\overline{a}\cdotp\overline{b}
\]
i.e. is equal to the standard inner product of vectors $\overline{a}=[a_1,\ldots,a_L]$ and $\overline{b}=[b_1,\ldots,b_L]$ in $\mathbb{R}^{L}$. The vectors
$\overline{a}$ corresponding to the diagonal matrices from the image of $d\mu|_{[\phi]}$ are tangent to $\Psi(\mathbb{P}(\mathcal{H}))$
at the point $\mu([\phi])$ satisfying (\ref{eq:condition1-1-1}). Therefore, if $\xi=[\xi_{1},\ldots,\xi_{L}]$ is a vector
perpendicular to the plane given by (\ref{eq:condition1-1-1}) then the corresponding operator $X=X_{1}\otimes I\otimes\ldots\otimes I+\ldots+I\otimes I\otimes\ldots\otimes X_{L}$,
where $X_{l}=\mathrm{diag}\{\xi_{l},-\xi_{l}\}$ is the element of the Lie algebra of the isotropy group $\mathfrak{k}_{[\phi]}$.
Consequently $X\phi=\lambda\phi$, for some $\lambda$.
\end{proof}
\noindent The vector $v=[-1,1,\ldots,1]$ is perpendicular to the plane given by (\ref{eq:condition1-1-1}). The corresponding
operator $X$ reads
\begin{gather}\label{matrix-X}
X=X_{1}\otimes I\otimes\ldots\otimes I+\ldots+I\otimes\ldots\otimes I\otimes X_{L},
\end{gather}
where $X_{1}=\mathrm{diag}\{-1,1\}$, $X_{2}=\ldots=X_{L}=\mathrm{diag}\{1,-1\}$. By proposition \ref{prop1} we need to consider
eigenspaces of $X$. We have the following:
\begin{proposition}
\label{The-spectrum}
The matrix $X$, defined by (\ref{matrix-X}), is a diagonal $2^{L}\times2^{L}$ matrix. The eigenvalues of $X$ are the integers chosen
from $-L$ to $L$ with the step $2$, that is $\sigma(X)=\{-L,\,-L+2,\ldots,L-2,L\}$. The multiplicity of eigenspace $\mathcal{H}_{-L+2k}$
is $\mathrm{dim}\mathcal{H}_{-L+2k}={L \choose k}$.
\end{proposition}
\begin{proof}
The matrices $X_{l}$ are diagonal and their spectra are $\sigma(X_{l})=\{-1,1\}$. Consequently, the matrix $X$ is also
diagonal and its eigenvalues are sums of eigenvalues of $X_{l}$'s. One can easily verify that the eigenvalues of $X$ belong
to the set $\sigma(X)=\{-L,\,-L+2,\ldots,L-2,L\}$. To see this note that the eigenvalue $-L+2k$ arises as a sum of eigenvalues
of $\{X_{l}\}_{l=1}^{L}$, where $k$ out of $L$ eigenvalues of $X_{l}$'s are positive ($+1$) and $L-k$ negative ($-1$).
Therefore the multiplicity of $\mathcal{H}_{-L+2k}$ is $\mathrm{dim}\mathcal{H}_{-L+2k}={L \choose k}$.
\end{proof}
\noindent As a direct consequence of proposition \ref{The-spectrum} we need to consider $L+1$ eigenspaces of $X$. In the following we describe the structure of these spaces and show that only one of them, that is, $\mathcal{H}_{-L+2}$ contains states
$[\phi]$ for which $\mu([\phi])$ consists of diagonal matrices whose diagonal elements satisfy (\ref{eq:condition1-1-1}).
The result is obtained in two steps. First in proposition \ref{states-form} we determine $\mathcal{H}_{-L+2k}$ and show that condition
(\ref{eq:condition1-1-1}) is not satisfied for $k\in\{0,\ldots,L\}\setminus\{1\}$. Then in proposition \ref{k=1} we prove
that for $\mathcal{H}_{-L+2}$ condition (\ref{eq:condition1-1-1}) is satisfied.

Denote by $D_{k}^{L}$ the subspace of $\left(\mathbb{C}^{2}\right)^{\otimes L}$ spanned by separable states of $L$ qubits
such that $k$ out of $L$ qubits are in the ground state $\ket 0$ and the remaining $L-k$ qubits are in the excited state
$\ket 1$, for example, $D_{2}^{3}=\mathrm{Span}_{\mathbb{C}}\{|001\rangle,|010\rangle,|100\rangle\}$. Assume that $D_{k}^{L}=\{0\}$
if $k>L$.
\begin{proposition}
\label{states-form}States which belong to eigenspace $\mathcal{H}_{-L+2k}$ are of the form $\phi=p_{0}\ket 0\otimes\psi_{1}+p_{1}\ket 1\otimes\psi_{2}$,
where $\psi_{1}\in D_{k}^{L-1}$ and $\psi_{2}\in D_{k-1}^{L-1}$. The reduced one-qubit density matrices for any $\phi\in\mathcal{H}_{-L+2k}$
are diagonal. For any $\phi\in\mathcal{H}_{-L+2k}$ condition (\ref{eq:condition1-1-1}) is equivalent to $(L-k-1)||\phi||^{2}=L-2$. \end{proposition}
\begin{proof}
We first determine vectors spanning eigenspace $\mathcal{H}_{-L+2k}$ of $X$. The eigenvalue $-L+2k$ arises as the sum
of eigenvalues of $\{X_{l}\}_{l=1}^{L}$ with $k$ out of $L$ eigenvalues equal to $(+1)$ and $L-k$ equal to $(-1)$.
Note, however, that matrix $X_{1}=\mathrm{diag}\{-1,1\}$, whereas $X_{2}=\ldots=X_{L}=\mathrm{diag}\{1,-1\}$. Therefore,
the eigenvalue $-L+2k$ corresponds to separable states with either the first qubit in the ground state $\ket 0$ and the
remaining $L-1$ qubits in a state from $D_{k}^{L-1}$or with the first qubit in the excited state $\ket 1$ and the remaining
$L-1$ qubits in a state from $D_{k-1}^{L-1}$. Thus the generic state belonging to $\mathcal{H}_{-L+2k}$ can be written
as
\begin{gather*}
\phi=\sum_{l=1}^{{L-1 \choose k}}a_{l}\ket 0\otimes e_{l}+\sum_{l=1}^{{L-1 \choose k-1}}b_{l}\ket 1\otimes f_{l},
\end{gather*}
where $\{e_{i}\}$ and $\{f_{i}\}$ are the separable states spanning $D_{k}^{L-1}$ and $D_{k-1}^{L-1}$ respectively. It
is straightforward to see that the reduced one-qubit density matrices of $\phi$ are diagonal. The first one is of the form:

\[
\rho_{1}([\phi])-\frac{1}{2}I=\left(\begin{array}{cc}
\sum_{l=1}^{{L-1 \choose k}}|a_{l}|^{2}-\frac{1}{2} & 0\\
0 & \sum_{l=1}^{{L-1 \choose k-1}}|b_{l}|^{2}-\frac{1}{2}
\end{array}\right),
\]
and hence
\begin{gather}
\lambda_{1}=\sum_{l=1}^{{L-1 \choose k-1}}|b_{l}|^{2}-\frac{1}{2}.
\label{eq:lambda1}
\end{gather}
We now show that (\ref{eq:condition1-1-1}) is equivalent to $(L-k-1)||\phi||^{2}=L-2$. As matrices $\rho_{l}([\phi])$
are diagonal, the $\left(\rho_{l}([\phi])\right)_{11}$ entry of each $\rho_{l}$, which is equal to $\lambda_{l}+\frac{1}{2}$,
is the sum of $|b_{j}|^{2}$ and $|a_{j}|^{2}$ coefficients corresponding to vectors with the $l$-th qubit in the excited
state $\ket 1$. Consequently, in the sum $\sum_{l=2}^{L}\left(\rho_{l}([\phi])\right)_{11}$, each $|b_{j}|^{2}$ coefficient
occurs $L-k$ times and each $|a_{j}|^{2}$ coefficient $L-k-1$ times. Therefore,
\begin{gather}
\sum_{l=2}^{L}\lambda_{l}=(L-k-1)\sum_{i=1}^{{L-1 \choose k}}|a_{i}|^{2}+(L-k)\sum_{i=1}^{{L-1 \choose k-1}}|b_{i}|^{2}-\frac{1}{2}(L-1)=\label{eq:sum}\\
=(L-k)||\phi||^{2}-\sum_{i=1}^{{L-1 \choose k}}|a_{i}|^{2}-\frac{1}{2}(L-1).\nonumber
\end{gather}
Using (\ref{eq:lambda1})

\begin{gather*}
-\lambda_{1}+\sum_{i=2}^{L}\lambda_{i}=(L-k-1)||\phi||^{2}-\frac{1}{2}L+1.
\end{gather*}
Hence equation (\ref{eq:condition1-1-1}) reads

\begin{gather*}
(L-k-1)||\phi||^{2}=L-2.
\end{gather*}

\end{proof}

\noindent Using fact \ref{states-form} one easily finds that for normalized state, i.e. when $||\phi||=1$ condition (\ref{eq:condition1-1-1})
can be satisfied only when $k=1$. The following proposition ensures that indeed this is the case.
\begin{proposition}
\label{k=1}The reduced one-qubit density matrices of states $\phi\in\mathcal{H}_{-L+2}$ are diagonal and satisfy
condition (\ref{eq:condition1-1-1}).
\end{proposition}
\begin{proof}
The eigenspace $\mathcal{H}_{-L+2}$ is $L$-dimensional. Any vector $\phi\in\mathcal{H}_{-L+2}$ can be written as

\begin{gather}
\phi=c_{1}\ket 1\otimes\ket{1\ldots1}+c_{2}\ket 0\otimes\ket{01\ldots1}+c_{3}\ket 0\otimes\ket{101\ldots1}+\ldots+c_{L}\ket 0\otimes\ket{1\ldots110},\label{eq:almost W}
\end{gather}
that is, $\phi$ is a linear combination of separable state where all qubits are in the excited state and the states for
which the first and one additional qubits are in the ground state (while other are in the excited state). Assume that $\phi$
is normalized, i.e. $\sum_{k=1}^{L}|c_{i}|^{2}=1$. It is straightforward to calculate

\begin{gather*}
\rho_{l}([\phi])-\frac{1}{2}I=\left(\begin{array}{cc}
|c_{l}|^{2}-\frac{1}{2} & 0\\
0 & -\frac{1}{2}+\sum_{k\neq l}^{L}|c_{k}|^{2}
\end{array}\right),\,\,\, i\in\{2,\ldots,L\},\\
\rho_{1}([\phi])-\frac{1}{2}I=\left(\begin{array}{cc}
-\frac{1}{2}+\sum_{k=2}^{L}|c_{k}|^{2} & 0\\
0 & |c_{1}|^{2}-\frac{1}{2}
\end{array}\right).
\end{gather*}
Note that we can assume that $|c_{1}|^{2}\geq\sum_{k=2}^{L}|c_{k}|^{2}$. This means that for all $i\in\{1,\ldots,L\}$ we
have$\left(\rho_{l}([\phi])-\frac{1}{2}I\right)_{11}\leq0$ as required. It is also easy to see that condition (\ref{eq:condition1-1-1})
is equivalent to

\begin{gather*}
\sum_{k=2}^{L}|c_{k}|^{2}=\sum_{k=2}^{L}|c_{k}|^{2},
\end{gather*}
and is satisfied.
\end{proof}
\noindent By proposition \ref{k=1}, states mapped by $\Psi$ onto $\alpha$'s satisfying (\ref{eq:condition1-1-1}) belong
to $K$-orbits through $\phi\in\mathcal{H}_{-L+2}$. On the other hand states $\phi\in\mathcal{H}_{-L+2}$ are $K^{\mathbb{C}}$-equivalent
to $L$-qubit $W$ state, where $K^\mathbb{C}=SL(2,\mathbb{C})^{\times L}$ is complexification of $K=SU(2)^{\times L}$ and
\begin{gather*}
[W]=\ket{01\ldots1}+\ket{101\ldots1}+\ldots+\ket{1\ldots10}.
\end{gather*}
This can be easily seen by changing $\ket 0\leftrightarrow\ket 1$ on the first qubit of (\ref{eq:almost W}). It was shown
in ref.\cite{SWK13} that the variety $\overline{K^{\mathbb{C}}.[W]}$ is {\it spherical}, i.e. reduced spaces stemming from the restriction $\Psi|_{\overline{K^{\mathbb{C}}.[W]}}$
are zero-dimensional. Therefore:
\begin{theorem}
Let $\alpha\in\Psi(\mathbb{P}(\mathcal{H}))$ be such that at least one of the inequalities $\left(\frac{1}{2}-\lambda_{l}\right)\leq\sum_{j\neq l}\left(\frac{1}{2}-\lambda_{j}\right)$
is equality. Then $\mathrm{dim}M_{\alpha}=0$.
\end{theorem}

\section{Calculation of $\mathrm{dim}M_{\alpha}$ for case 3}
\label{sec:case3}

As we showed in section \ref{sec:The-reduced-spaces}, for points $\alpha$ in the interior of $\Psi(\mathbb{P}(\mathcal{H}))$
the dimension of the reduced space is

\begin{gather*}
\mathrm{dim}M_{\alpha_{gen}}=2^{L+1}-4L-2.
\end{gather*}
In the following we show that for $\alpha=(\alpha_{1},\ldots,\alpha_{L})\in\Psi(\mathbb{P}(\mathcal{H}))$ with matrices
$\alpha_{1}=\ldots=\alpha_{k}=0$
\begin{gather}
\mathrm{dim}M_{\alpha}=\left(2^{L+1}-4L-2\right)-2k=\mathrm{dim}M_{\alpha_{gen}}-2k.\label{eq:case3}
\end{gather}
This means that the dimension $\mathrm{dim}M_{\alpha}$ drops by $2$ every time one of the reduced density matrices becomes
maximally mixed. The argument for this is based on the existence of {\it stable states} which we discuss first.

\subsection{Stable states}

In the following we assume that $\tilde{K}$ is any compact semisimple group acting in the symplectic way on the complex
projective space $\mathbb{P}(\mathcal{H})$, where $\mathcal{H}$ can be, for example, the Hilbert space of $L$ qubits.
We will denote by $\tilde{\mu}$ the corresponding momentum map and assume that $\tilde{\mu}^{-1}(0)\neq\emptyset$. Let
$\tilde{G}=\tilde{K}^{\mathbb{C}}$ be the complexification of $\tilde{K}$. Following ref. \cite{HH96} we denote
\begin{gather*}
X(\tilde{\mu})=\{[\phi]\in\mathbb{P}(\mathcal{H}):\,\overline{\tilde{G}.[\phi]}\cap\tilde{\mu}^{-1}(0)\neq\emptyset\}.
\end{gather*}
It is known that the set $X(\tilde{\mu})$ is an open dense subset of $\mathbb{P}(\mathcal{H})$. Moreover the set $G.\tilde{\mu}^{-1}(0)\subset X(\tilde{\mu})$
is also an open dense subset of $\mathbb{P}(\mathcal{H})$ (see ref.\cite{M77}). We will use the following terminology, typical for geometric invariant
theory:
\begin{enumerate}
\item $[\phi]$ is unstable iff $[\phi]\notin X(\tilde{\mu})$,
\item $[\phi]$ is semistable iff $[\phi]\in X(\tilde{\mu})$.
\end{enumerate}
Among semistable states we distinguish the class of stable states. By definition, a semistable state $[\phi]$ is stable if and only if $\tilde{\mu}([\phi])=0$ and $\mathrm{dim}\tilde{K}.[\phi]=\mathrm{dim}\tilde{K}$. Note that since for $[\phi]\in\tilde{\mu}^{-1}(0)$ one has $\mathrm{dim}\tilde{K}^{\mathbb{C}}.[\phi]=2\mathrm{dim}\tilde{K}.[\phi]$ (see ref.\cite{Kirwan-thesis}), the condition $\mathrm{dim}\tilde{K}.[\phi]=\mathrm{dim}\tilde{K}$ can be phrased as $\mathrm{dim}\tilde{G}.[\phi]=\mathrm{dim}\tilde{G}$. Remarkably, the existence of a stable state implies that almost all semistable states are stable, in particular almost all states in $\tilde{\mu}^{-1}(0)$ are stable \cite{HH96}. Note that since $\tilde{G}.\tilde{\mu}^{-1}(0)$
is open and dense in $\mathbb{P}(\mathcal{H})$ and a generic $\tilde{G}$-orbit in $\tilde{G}.\tilde{\mu}^{-1}(0)$ has
dimension $\mathrm{dim}\tilde{G}$ we get

\begin{gather}
\mathrm{dim}\mathbb{P}(\mathcal{H})=\mathrm{dim}\tilde{G}.\tilde{\mu}^{-1}(0)=\mathrm{dim}\tilde{G}+\mathrm{dim}\left(\tilde{G}.\tilde{\mu}^{-1}(0)\right)/\tilde{G}.\label{eq:relation}
\end{gather}
One of the central results in the geometric invariant theory reads \cite{Kirwan-thesis}

\begin{gather}
\left(\tilde{G}.\tilde{\mu}^{-1}(0)\right)/\tilde{G}=\tilde{\mu}^{-1}(0)/\tilde{K}.\label{eq:categorical}
\end{gather}
Hence, under the assumption of stable states existence and using (\ref{eq:relation}) and (\ref{eq:categorical}) we get

\begin{gather}
\tilde{\mu}^{-1}(0)/\tilde{K}=\mathrm{dim}\mathbb{P}(\mathcal{H})-2\mathrm{dim}\tilde{K}.\label{eq:dim-categorical}
\end{gather}
As we will see formula (\ref{eq:dim-categorical}) plays a major role in showing (\ref{eq:case3}).

\subsection{The strategy for showing (\ref{eq:case3})}

Let $K=K_{1}\times K_{2}$, where $K_{1}=SU(2)^{\times k}$ and $K_{2}=SU(2)^{\times(L-k)}$. We first consider the natural action
of $K_{1}$ on the first $k$ qubits in $\mathbb{P}(\mathcal{H})$, where $\mathcal{H}=\left(\mathbb{C}^{2}\right)^{\otimes L}$. The momentum map
$\mu_{1}$ for this action gives the first $k$ reduced density matrices. Therefore, $\mu_{1}^{-1}(0)$ consists of all states
with the first $k$ reduced density matrices maximally mixed, but no assumption is made on the spectra of the remaining $(L-k)$
matrices. In the following we assume that there exists a stable state for $K_{1}$-action on $\mathbb{P}(\mathcal{H})$ (see
lemma \ref{lemma-stable} for proof). Under this assumption and using formula (\ref{eq:dim-categorical}) the dimension of
$\mathrm{dim}\mu_{1}^{-1}(0)/K_{1}$ is

\begin{gather*}
\mathrm{dim}\mu_{1}^{-1}(0)/K_{1}=\mathrm{dim}\mathbb{P}(\mathcal{H})-2\mathrm{dim}K_{1}=2^{L+1}-6k-2.
\end{gather*}
Recall that $\mu_{1}^{-1}(0)/K_{1}$ is a stratified symplectic space and we consider the highest dimensional stratum which
is a symplectic manifold. Removing $K_{1}$ freedom does not affect $K_{2}$ action, i.e. the actions of $K_1$ and $K_2$ commute. Therefore, we can consider action of
$K_{2}$ on the highest dimensional stratum of $\mu_{1}^{-1}(0)/K_{1}$. The momentum map $\mu_{2}$ for $K_{2}$ action
on $\mu_{1}^{-1}(0)/K_{1}$ gives the remaining $L-k$ reduced density matrices. Moreover, using inequalities (\ref{nier1})
with $\lambda_{1}=\ldots=\lambda_{k}=0$ it is straightforward to see that the image of the corresponding map $\Psi_{2}$
is $L-k$ dimensional polytope. Using fact \ref{dim-orbit} and formula (\ref{eq:fiber-interior-1}), for a point inside
of this polytope, e.g. when $\lambda_{k+1},\ldots,\lambda_{L}\neq0$, the dimension of $\Psi_{2}$-fiber is

\begin{gather*}
\left(\left(\mathrm{dim}\mu_{1}^{-1}(0)/K_{1}\right)-\left(L-k\right)\right)-\mathrm{dim}K_{2}=\left(\left(2^{L+1}-6k-2\right)-\left(L-k\right)\right)-3(L-k)\\
=2^{L+1}-4L-2k-2.
\end{gather*}
But the $\Psi_{2}$-fiber is exactly the reduced space we look for, i.e. the one which corresponds to $\lambda_{1}=\ldots=\lambda_{k}=0$
and $\lambda_{k+1},\ldots,\lambda_{L}\neq0$. Therefore, as promised

\begin{gather*}
\mathrm{dim}M_{\alpha}=\mathrm{dim}M_{\alpha_{gen}}-2k=2^{L+1}-4L-2k-2.
\end{gather*}
In order to complete the above reasoning we now show that an appropriate stable state indeed exists.
\begin{lemma}
\label{lemma-stable}Let $K_{1}=SU(2)^{\times k}$, $k\leq L$. If $L\geq5$ then the $L$-qubit state
\begin{gather*}
[\phi]=\left(|0\ldots0\rangle+|1\ldots1\rangle\right)+\left(|110\ldots0\rangle+|001\ldots1\rangle\right)+\left(|1010\ldots0\rangle+|0101\ldots1\rangle\right)+\\
+\ldots+\left(|10\ldots01\rangle+|01\ldots10\rangle\right),
\end{gather*}
is $K_{1}$-stable. For $L=4$
\begin{gather*}
[\phi]=\alpha\left(|0000\rangle+|1111\rangle\right)+\left(|1100\rangle+|0011\rangle\right)+\left(|1010\rangle+|0101\rangle\right)+\\
+\left(|1001\rangle+|0110\rangle\right),\ \alpha\in\mathbb{R}\setminus\{1,-3\},
\end{gather*}
is $K_{1}$-stable.
\end{lemma}

\begin{proof}
We need to show that the first $k$ reduced density matrices of $[\phi]$ are maximally
mixed and that $\mathrm{dim}K_{1}.[\phi]=\mathrm{dim}K_{1}$. Note first that if the state $[\phi]$ is stable with respect
to $K=SU(2)^{\times L}$ action it is also stable with respect to $K_{1}\subset K$ action. Therefore, we will show that
$[\phi]$ is $K$-stable.

The state $[\phi]$ consists of $L\geq5$ pairs of separable states. In each pair the second vector is the first vector with
the swap $|0\rangle\leftrightarrow|1\rangle$ performed on every qubit. The first pair is GHZ state. In the remaining pairs
the first vector is such that the first and one additional qubits are in the excited state $\ket 1$ while the remaining
$L-2$ qubits are in the ground state $\ket 0$. The construction ensures that $\mu([\phi])=0$. What is left is to calculate
the dimension $\mathrm{dim}G.[\phi]$, where $G=K^{\mathbb{C}}=SL(2,\mathbb{C})^{\times L}$. This is equivalent to calculating
the dimension of the tangent space $T_{[\phi]}G.[\phi]$ which is generated by the action of Lie algebra $\mathfrak{g}=\mathfrak{sl}(2,\mathbb{C})^{\times L}$
on $[\phi]$. More precisely let

\begin{gather*}
E_{12}=\left(\begin{array}{cc}
0 & 1\\
0 & 0
\end{array}\right),\,\, E_{21}=\left(\begin{array}{cc}
0 & 0\\
1 & 0
\end{array}\right),\,\, H=\left(\begin{array}{cc}
1 & 0\\
0 & -1
\end{array}\right),
\end{gather*}
be the basis of $\mathfrak{sl}(2,\mathbb{C})$. Let
\begin{gather*}
E_{12}^{(l)}=\mathbb{I}\otimes\mathbb{I}\otimes\ldots\otimes E_{12}\otimes\mathbb{I}\otimes\ldots\otimes\mathbb{I},\\
E_{21}^{(l)}=\mathbb{I}\otimes\mathbb{I}\otimes\ldots\otimes E_{12}\otimes\mathbb{I}\otimes\ldots\otimes\mathbb{I},\\
H^{(l)}=\mathbb{I}\otimes\mathbb{I}\otimes\ldots\otimes H\otimes\mathbb{I}\otimes\ldots\otimes\mathbb{I}.
\end{gather*}
Our goal is to show that vectors $\{E_{12}^{(l)}\phi,E_{21}^{(l)}\phi,H^{(l)}\phi\}_{l=1}^{L}$ are linearly independent
and orthogonal to $\phi$. Denote by $\ket{00...0}\otimes\ket 1_{l}$ a separable state whose $l$-th qubit is in the excited
state $\ket 1$ and remaining qubits are in the ground state $\ket 0$ (e.g. a 5 - qubit state $\ket{00\ldots0}\otimes\ket 1_{2}=\ket{01000}$
and a 6 - qubit state $\ket{00\ldots0}\otimes\ket 1_{1}\otimes\ket 1_{4}\otimes\ket 1_{6}=\ket{100101}$). It is straightforward
to verify that
\begin{gather*}
E_{21}^{(1)}\phi=\ket{00\ldots0}\otimes\ket 1_{1}+\sum_{l'=2}^{L}\ket{11\ldots1}\otimes\ket 0_{l'},\\
E_{21}^{(l)}\phi=\ket{00\ldots0}\otimes\ket 1_{l}+\ket{11\ldots1}\otimes\ket 0_{1}+\sum_{l'\neq l}\ket{00\ldots0}\otimes\ket 1_{1}\otimes\ket 1_{l}\otimes\ket 1_{l'},\ l\geq2.
\end{gather*}
Moreover, the action of $E_{12}$ gives vectors that can be obtained from the above set of vectors by preforming the swap
operation on every qubit. The vectors obtained by the action of $E_{21}$ are linearly independent, because each vector
from this group contains a unique separable state of the form $\ket{00...0}\otimes\ket 1_{l}$ (an analogous argument can be
applied to the set of vectors obtained by the action of $E_{12}$). We next examine the linear independence of vectors from
groups $E_{12}^{(l)}$ and $E_{21}^{(l)}$. First, note that $E_{12}^{(l)}\phi,\ l\geq2$ are linearly independent from all
vectors $E_{21}^{(l)}\phi$, because they consist of separable states of the form $\ket{11\ldots1}\otimes\ket 0_{1}\otimes\ket 0_{l}\otimes\ket 0_{l'}$
that do not appear in the vectors $E_{21}^{(l)}\phi$ (it is not true in the 4 - qubit case, but we will return to this
problem later). The last thing to show is the linear independence of the vector $E_{12}^{(1)}\phi=\ket{11\ldots1}\otimes\ket 0_{1}+\sum_{l'=2}^{L}\ket{00\ldots0}\otimes\ket 1_{l'}$
from the vectors $E_{21}^{(l)}\phi,\ l\geq2$. Note that this vector is orthogonal to $E_{21}^{(1)}\phi$. Assume that $E_{12}^{(1)}\phi$ can be expressed as the linear combination of the remaining vectors

\begin{gather}
E_{12}^{(1)}\phi=\sum_{l=2}^{L}\lambda_{l}E_{21}^{(l)}\phi.\label{eq:linear}
\end{gather}
We will show that this leads to a contradiction. To this end, let us calculate the sum

\begin{gather*}
\sum_{l=2}^{L}E_{21}^{(l)}\phi=\sum_{l=2}^{L}\ket{00\ldots0}\otimes\ket 1_{l}+(L-1)\ket{11\ldots1}\otimes\ket 0_{1}+\sum_{l}\sum_{l'\neq l}\ket{00\ldots0}\otimes\ket 1_{1}\otimes\ket 1_{l}\otimes\ket 1_{l'}=\\ 
= E_{12}^{(1)}\phi+(L-2)\ket{11\ldots1}\otimes\ket 0_{1}+\sum_{l}\sum_{l'\neq l}\ket{00\ldots0}\otimes\ket 1_{1}\otimes\ket 1_{l}\otimes\ket 1_{l'}.
\end{gather*}
\noindent Using (\ref{eq:linear}) we get
\begin{equation}
\sum_{l=2}^{L}(1-\lambda_{l})E_{21}^{(l)}\phi=(L-2)\ket{11\ldots1}\otimes\ket 0_{1}+\sum_{l}\sum_{l'\neq l}\ket{00\ldots0}\otimes\ket 1_{1}\otimes\ket 1_{l}\otimes\ket 1_{l'} 
\label{eq:sum1}
\end{equation}
\noindent On the other hand,
\begin{gather}
\sum_{l=2}^{L}(1-\lambda_{l})E_{21}^{(l)}\phi=\sum_{l=2}^{L}(1-\lambda_{l})\ket{00\ldots0}\otimes\ket 1_{l}+
\sum_{l=2}^{L}(1-\lambda_{l})\ket{11\ldots1}\otimes\ket 0_{1}+ \nonumber \\
+\sum_{l}(1-\lambda_{l})\sum_{l'\neq l}\ket{00\ldots0}\otimes\ket 1_{1}\otimes\ket 1_{l}\otimes\ket 1_{l'} .
\label{eq:sum2}
\end{gather}
\noindent Thus, comparing the coefficients by $\ket{11\ldots1}\otimes\ket 0_{1}$ and $\ket{00\ldots0}\otimes\ket 1_{l}$ in equations (\ref{eq:sum1}) and (\ref{eq:sum2}),  we get
\begin{gather}\label{eq:cond1}
\sum_{l=2}^{L}(1-\lambda_{l})=L-2,\,\,\,\forall_l\ 1-\lambda_{l}=0
\end{gather}
This is a contradiction, because the first equation (\ref{eq:cond1}) is equivalent to $\sum_{l=2}^{L}\lambda_{l}=1$, which cannot be satistied by 
$\lambda_{l}=1$ for all $l$ (which is implied by the second equation (\ref{eq:cond1})).
 Clearly for any $l\in\{1,\ldots,L\}$ we also have $\bk{E_{12}^{(l)}}{\phi}=0=\bk{E_{21}^{(l)}\phi}{\phi}$.
We are left with vectors $H^{(l)}\phi$. It is straightforward to see that $\bk{H^{(l)}\phi}{\phi}=\bk{H^{(l)}\phi}{E_{21}^{(l)}\phi}=\bk{H^{(l)}\phi}{E_{12}^{(l)}\phi}=0$.
The matrix of coefficients for $\{H^{(l)}\phi\}$ is given by
\begin{gather*}
C^{\prime}=\left(\begin{array}{ccccccc}
1 & 1 & . & . & . & . & 1\\
1 & 1 & -1 & . & . & . & -1\\
1 & -1 & 1 & -1 & ... & -1\\
1 & -1 & -1 & 1 & -1 & ... & -1\\
. &  &  & . &  &  & .\\
. &  &  & . &  &  & .\\
1 & -1 & -1 & . & . & -1 & 1
\end{array}\right)
\end{gather*}
in the basis $\{(|00...0\rangle-|11...1\rangle),\,(-|110...0\rangle+|001...1\rangle),\,(-|1010...0\rangle+|0101...1\rangle),\ldots,(-|10...01\rangle+|01...10\rangle)\}$.
By direct calculation one checks that $\mathrm{det}C^{\prime}\neq0$. Therefore the dimension of $G.[\phi]$ is equal to
the dimension of $G$, which is $6L$ and $\phi$ is $K$-stable. In the case of 4 qubits we need to consider a slightly
different state

\begin{gather*}
[\phi]=\alpha\left(|0000\rangle+|1111\rangle\right)+\left(|1100\rangle+|0011\rangle\right)+\left(|1010\rangle+|0101\rangle\right)+\\
+\left(|1001\rangle+|0110\rangle\right),\ \alpha\in\mathbb{R}\setminus\{1,-3\}.
\end{gather*}
It can be shown by similar calculation that this state is $K$ - stable.
\end{proof}

\section{Summary}
Given spectra of one-qubit reduced density matrices, we found the formula for the minimal number of polynomials needed to decide LU equivalence of $L$-qubit pure states. As we showed this number is the same for spectra belonging to the interior of the polytope $\Psi(\mathbb{P}(\mathcal{H}))$. This is not the case on the boundary where the behaviour of $\mathrm{dim}M_\alpha$ is not uniform. In particular, for a large part of the boundary of $\Psi(\mathbb{P}(\mathcal{H}))$ we have $\mathrm{dim}M_{\alpha_{b}}=0$. We also observed
that for $\alpha_{b}\in\Psi(\mathbb{P}(\mathcal{H}))$ corresponding to the $\left\{ \rho_{l}([\phi])\right\} _{l=1}^{L}$ such that $k$ matrices are
maximally mixed $\mathrm{dim}M_{\alpha_{b}}=\mathrm{dim}M_{\alpha_{gen}}-2k$.

The methods used in this paper can be in principle applied to $L$-particle systems with an arbitrary finite-dimensional one-particle Hilbert spaces. The argument for points in the interior of $\Psi(\mathbb{P}(\mathcal{H}))$ can be used {\it mutatis mutandis} in this case. We note, however, that inequalities describing the polytope $\Psi(\mathbb{P}(\mathcal{H}))$ are much more complicated when $\mathcal{H}\neq\left(\mathbb{C}^2\right)^{\otimes L}$ (see ref. \cite{K04}) and therefore the problem for boundary points is of higher computational complexity. Nevertheless, one should expect that, similarly to the qubit case, there is a large part of the boundary characterized by $\mathrm{dim}M_\alpha = 0$.

\begin{acknowledgments}
We would like to thank Marek Ku\'s for many discussions on the geometric aspects of quantum correlations and his continuous
encouragement. This work is supported by Polish Ministry of Science and Higher Education Iuventus Plus grant no. IP2011048471.
\end{acknowledgments}

\appendix*
\section{Vertices of the polytope $\Psi(\mathbb{P}(\mathcal{H}))$ for $L$ qubits}

In order to find vertices of the polytope $\Psi(\mathbb{P}(\mathcal{H}))$, it is more convenient to view the inequalities
describing the polytope in terms of the minimal eigenvalues of the reduced one-qubit density matrices (remembering that the
shifted spectra are given by $\lambda_{i}=\frac{1}{2}-p_{i}$). As proven in \cite{HSS03} these inequalities read:
\begin{subequations}
\begin{eqnarray}
p_{i}\leq\sum_{j\neq i}p_{j}\,,\, i=1,\ldots,L\,,\label{eq:ntrivineq} \\
p_{i}\geq0\,,\, i=1,\ldots,L\,,\label{eq:trivineq1} \\
p_{i}\leq\frac{1}{2}\,,\, i=1,\ldots,L\,,\label{eq:trivineq2}
\end{eqnarray}
\end{subequations}
where each $p_{i}$ denotes the minimum eigenvalue of the reduced one-qubit density matrix describing $i$-th qubit, $\rho_{i}$.
The general algorithm of finding the vertices of the polytope given by a set of inequalities is to choose the set of $L$
of the inequalities, write them as a set of equations and check whether there exists a unique solution \cite{plytopes}.
If obtained solution satisfies the remaining $2L$ inequalities then it defines a vertex of $\Psi(\mathbb{P}(\mathcal{H}))$.
The following fact describes the structure of vertices of $\Psi(\mathbb{P}(\mathcal{H}))$
\begin{proposition}
\label{vertices}Vertices of the polytope $\Psi(\mathbb{P}(\mathcal{H}))$ in the case of $L$ qubits are given by equations
\begin{equation}
p_{i}=0\,\text{or }p_{i}=\frac{1}{2},\, i=1,\ldots,L\label{eq:vertices}
\end{equation}
such that the number of indecies $i$ for which $p_i=\frac{1}{2}$ belongs to $\{0,2,3,...,L\}$, that is,  $|\{i:p_{i}=\frac{1}{2}\}|\in\{0,2,3,...,L\}$.
\end{proposition}
\begin{proof}
First note that each vertex of this form can be obtained by picking $|\{i:p_{i}=\frac{1}{2}\}|$ equations from \eqref{eq:trivineq2}
and $L-|\{i:p_{i}=\frac{1}{2}\}|$ linearly independent equations from \eqref{eq:trivineq1}. One easily checks that in each case $2L$ remaining
inequalities are trivially satisfied. Thus, points described by conditions \eqref{eq:vertices} are indeed vertices of $\Psi(\mathbb{P}(\mathcal{H}))$.
A single exceptional case occurs when $|\{i:p_{i}=\frac{1}{2}\}|=1$, e.g. $p_{1}=\frac{1}{2}$. Then by (\ref{eq:ntrivineq})
we get $p_{1}=\frac{1}{2}\leq0$, which is a contradiction. We will now show that these are all solutions. We achieve our
goal by considering all remaining possibilities of choosing $L$ out of $3L$ inequalities. Before we proceed let us
introduce some useful notation. Any choice of $L$ out of $3L$ inequalities (and turning them into equalities) is uniquely
given when the following three auxiliary sets are specified:
\[I_{1}=\{i:i^\mathrm{th}\,\text{inequality of the form \eqref{eq:ntrivineq} was chosen}\},\]
\[I_{2}=\{i:i^\mathrm{th}\,\text{inequality of the form \eqref{eq:trivineq1} was chosen}\},\]
\[I_{3}=\{i:i^\mathrm{th}\,\text{inequality of the form \eqref{eq:trivineq2} was chosen}\}.\]
We have $\left|I_{1}\right|+\left|I_{2}\right|+\left|I_{3}\right|=L$, where $\left|I\right|$ denotes the number of elements
of the finite set $I$. We have already covered cases when $I_{1}=\emptyset$. What remains to be checked are three possibilities:
$I_{2}=\emptyset$ or $I_{3}=\emptyset$ and the case when each $I_{i}$ is non-empty. Firstly, let us consider the case
when $I_{3}=\emptyset$. If additionally $I_{2}=\emptyset$, one can check by direct calculations that the set of $L$ equations
from $I_{1}$ gives $p_{i}=0$ for all $i$. Next, if the set $I_{2}$ is non-empty, i.e. we choose some number of $p_{i}$'s
equal to zero and combine these conditions with the equations from $I_{1}$, our problem either reduces to a problem analogous
to the previous case with $I_{2}=\emptyset$, or there exists such $i\in I_{1}$ that reads $p_{i}=\sum_{j\neq i}p_{j}$ with
$p_{i}=0$. Because all $p_{i}$'s are positive or equal to zero, we obtain in both cases that $p_{j}=0$ for all $j$. One
of the last things to check is the case when the only empty set is $I_{2}$, i.e. $I_{2}=\emptyset$, $I_{3}\neq\emptyset$
and $I_{1}\neq\emptyset$. If, in addition $\left|I_{3}\right|=k>1$, there exists such $i\in I_{1}$ that reads $p_{i}=\sum_{j\neq i}p_{j}=\sum_{j\notin I_{3}}p_{j}+\frac{k}{2}$.
This implies that $p_{i}>\frac{1}{2}$, which is a contradiction. Now, if $\left|I_{3}\right|=1$ from the same equation
we get that for all $j\notin I_{3}$ $p_{j}=0$ and $p_{i}=\frac{1}{2}$. This is a contradiction, because we assumed that
$\left|I_{3}\right|=1$. This argument remains also true in the case when all sets are non-empty.
\end{proof}
By the above fact, the number of the vertices of the polytope $\Psi(\mathbb{P}(\mathcal{H}))$ for $L$ qubits is
the number of ways to place $k$ out of $L$ $p_{i}$'s equal to $\frac{1}{2}$ and the remaining
$p_{i}$s equal to zero on $L$ places:

\[
V=\binom{L}{0}+\binom{L}{2}+...+\binom{L}{L}=\sum_{k=0}^{N}\binom{L}{k}-\binom{L}{1}=2^{L}-L.
\]
Moreover, to find the $L-1$ dimensional faces of the polytope, one has to change one of the inequalities (\ref{eq:ntrivineq}),
(\ref{eq:trivineq1}), (\ref{eq:trivineq2}) into equality and find vertices that satisfy this condition (the minimal
number of vertices sufficient to span such face is $L$). In this way, for $L\geq4$ qubits, one obtains that there
are $3L$ such faces.

\subsection{The 4 qubits example}

The four-qubit polytope is a 4-dimensional convex polytope spanned by 12 vertices (in terms of $\lambda$'s):

\begin{table}[H]
\begin{ruledtabular}
\begin{tabular}{ccc}
$|\{i:p_{i}=\frac{1}{2}\}|$ & Vertices in terms of $\lambda_i=\frac{1}{2}-p_i$ \\ \hline \hline
0 & $\mathrm{v}_{\mathrm{SEP}}=\{\textrm{diag}(-\frac{1}{2},\frac{1}{2}),\,\textrm{diag}(-\frac{1}{2},\frac{1}{2}),\,\textrm{diag}(-\frac{1}{2},\frac{1}{2}),\,\textrm{diag}(-\frac{1}{2},\frac{1}{2})\}$ \\ \hline
\multirow{6}{*}{2} & $\mathrm{v}_{\mathrm{B1}}=\{\textrm{diag}(0,0),\,\textrm{diag}(0,0),\,\textrm{diag}(-\frac{1}{2},\frac{1}{2}),\,\textrm{diag}(-\frac{1}{2},\frac{1}{2})\}$ \\
 &  $\mathrm{v}_{\mathrm{B2}}=\{\textrm{diag}(0,0),\,\textrm{diag}(-\frac{1}{2},\frac{1}{2}),\,\textrm{diag}(0,0),\,\textrm{diag}(-\frac{1}{2},\frac{1}{2})\}$ \\
 &  $\mathrm{v}_{\mathrm{B3}}=\{\textrm{diag}(0,0),\,\textrm{diag}(-\frac{1}{2},\frac{1}{2}),\,\textrm{diag}(-\frac{1}{2},\frac{1}{2}),\,\textrm{diag}(0,0)\}$ \\
 &  $\mathrm{v}_{\mathrm{B4}}=\{\textrm{diag}(-\frac{1}{2},\frac{1}{2}),\,\textrm{diag}(0,0),\,\textrm{diag}(0,0),\,\textrm{diag}(-\frac{1}{2},\frac{1}{2})\}$ \\
 &  $\mathrm{v}_{\mathrm{B5}}=\{\textrm{diag}(-\frac{1}{2},\frac{1}{2}),\,\textrm{diag}(0,0),\,\textrm{diag}(-\frac{1}{2},\frac{1}{2}),\,\textrm{diag}(0,0)\}$ \\
 &  $\mathrm{v}_{\mathrm{B6}}=\{\textrm{diag}(-\frac{1}{2},\frac{1}{2}),\,\textrm{diag}(-\frac{1}{2},\frac{1}{2}),\,\textrm{diag}(0,0),\,\textrm{diag}(0,0)\}$ \\ \hline
\multirow{4}{*}{3} & $\mathrm{v}_4=\{\textrm{diag}(0,0),\,\textrm{diag}(0,0),\,\textrm{diag}(0,0),\,\textrm{diag}(-\frac{1}{2},\frac{1}{2})\}$ \\
 &  $\mathrm{v}_3=\{\textrm{diag}(0,0),\,\textrm{diag}(0,0),\,\textrm{diag}(-\frac{1}{2},\frac{1}{2}),\,\textrm{diag}(0,0)\}$ \\
 &  $\mathrm{v}_2=\{\textrm{diag}(0,0),\,\textrm{diag}(-\frac{1}{2},\frac{1}{2}),\,\textrm{diag}(0,0),\,\textrm{diag}(0,0)\}$ \\
 &  $\mathrm{v}_1=\{\textrm{diag}(-\frac{1}{2},\frac{1}{2}),\,\textrm{diag}(0,0),\,\textrm{diag}(0,0),\,\textrm{diag}(0,0)\}$ \\
\hline
4 & $\mathrm{v}_{\mathrm{GHZ}}=\{\textrm{diag}(0,0),\,\textrm{diag}(0,0),\,\textrm{diag}(0,0),\,\textrm{diag}(0,0)\}$
\end{tabular}
\end{ruledtabular}
\caption{The vertices of the four-qubit polytope $\Psi(\mathbb{P}(\mathcal{H}))$.}
\label{vertices}
\end{table}

The 3-dimensional faces can be divided into three groups obtained by: (1) changing one inequality from (\ref{eq:ntrivineq})
into equality: the face is spanned by $\mathrm{v}_{\mathrm{SEP}}$ and three $\mathrm{v}_{\mathrm{B}_j}$ vertices, (2) choosing one $\lambda_{i}=0$:
the face is spanned by three $\mathrm{v}_{\mathrm{B}_j}$ vertices, three $\mathrm{v}_{j}$ vertices and $\mathrm{v}_{\mathrm{GHZ}}$, (3) choosing
one $\lambda_{i}=\frac{1}{2}$: the face is spanned by $\mathrm{v}_{\mathrm{SEP}}$, three $\mathrm{v}_{\mathrm{B}_j}$ vertices and one $\mathrm{v}_{j}$ vertex.


\begin{thebibliography}{99}
\bibitem{B72}Bredon, G. E.,
\newblock Introduction to compact transformation groups,
\newblock {Pure and Applied Math.} \textbf{46}, Academic Press, 1972.

\bibitem{CDKW12}Christandl, M., Doran, B., Kousidis, S., Walter, M.,
\newblock Eigenvalue Distributions of Reduced Density Matrices.
\newblock  arXiv:1204.0741, 2012.

\bibitem{plytopes}Gr\"{u}nbaum, B. Kaibel, Volker; Klee, Victor; Ziegler, G\"{u}nter M., eds.,
\newblock Convex polytopes (2nd ed.),
\newblock  New York \& London, 2003.

\bibitem{GS84}Guillemin, V., Sternberg, S.,
\newblock Convexity properties of the moment mapping,
\newblock  {Invent. Math.} \textbf{67}, 491513, 1982.

\bibitem{HH96}Heinzner, P., Huckleberry, A.,
\newblock K\"{a}hlerian potentials and convexity properties of the moment map,
\newblock  {Invent. Math.} \textbf{126}, 6584, 1996.

\bibitem{HSS03}Higuchi, A., Sudbery, A., and Szulc, J.,
\newblock One-qubit reduced states of a pure many-qubit state: polygon inequalities,
\newblock  {Phys. Rev. Lett.} \textbf{90}, 107902, 2003.

\bibitem{K84}Kirwan, F. C.,
\newblock Convexity properties of the moment mapping,
\newblock {Invent. Math.} \textbf{77}, 547552, 1984

\bibitem{Kirwan-thesis}Kirwan, F. C.,
\newblock Cohomology of Quotients in Symplectic and Algebraic Geometry,
\newblock  {Mathematical Notes} \textbf{31}, Princeton, NJ: Princeton University Press, 1982.

\bibitem{K04}Klyachko, A.
\newblock Quantum marginal problem and representation of the symmetry group,
\newblock {arXiv:quant-ph/0409113}

\bibitem{Kraus1}Kraus, B.,
\newblock Local unitary equivalence of multipartite pure states,
\newblock  {Phys. Rev. Lett.} \textbf{104}(2), 020504, 2010.

\bibitem{Kraus2}Kraus, B.,
\newblock Local unitary equivalence and entanglement of multipartite pure states,
\newblock  arXiv 1005.5295, 2010.

\bibitem{MW99}Meinrenken, E., Woodward, C.,
\newblock Moduli spaces of flat connections on 2-manifolds, cobordism, and Witten’s volume formulas,
\newblock  {Advances in geometry}, Progr. Math. \textbf{172}, 271–295, 1999.

\bibitem{ZM12}Molladavoudi, S.
\newblock On the Symplectic Reduced Space of Three-Qubit Pure States,
\newblock arXiv 1005.5295, 2012.


\bibitem{M77}Mumford, D.
\newblock Stability of projective varieties,
\newblock  {L'Enseignement Math\'{e}matique}, 1977.

\bibitem{SHK11}Sawicki, A., Huckleberry, A., Ku\'{s}, M.,
\newblock Symplectic geometry of entanglement,
\newblock  {Comm. Math. Phys.} \textbf{305}, 441–468, 2011.

\bibitem{SK11}Sawicki, A., Ku\'{s}, M.,
\newblock Geometry of the local equivalence of states,
\newblock  {J. Phys. A: Math. Theor.} \textbf{44,} 495301, 2011.

\bibitem{SWK13} Sawicki, A., Walter, M., Ku\'{s}, M.
\newblock When is a pure state of three qubits determined by its single-partite reduced density matrices?,
\newblock  {J. Phys. A: Math. Theor.} \textbf{46,} 055304, 2013.

\bibitem{WDGC12}Walter, M., Doran, B., Gross, D., Christandl, M.,
\newblock Entanglement Polytopes,
\newblock  arXiv:1208.0365, 2012.

\end{thebibliography}
\end{document}